\documentclass[letterpaper, 10 pt, conference]{ieeeconf}  % Comment this line out
\IEEEoverridecommandlockouts         % This command is only
\usepackage{balance} %even columns on last page

\usepackage{xcolor,calc}

\usepackage[doi=false,isbn=false,defernumbers=true,style=ieee,backend=bibtex,maxbibnames=1000]{biblatex} % use bibtex as backend instead of biber for arXiv.
\bibliography{bibliography.bib}

\usepackage[pdftex]{graphicx}
\usepackage{amsmath} % assumes amsmath package installed
\usepackage{amssymb}  % assumes amsmath package installed
\usepackage{color}
\usepackage{enumerate}
\usepackage{multirow}
\usepackage{mathtools}
\usepackage{booktabs}
\usepackage{epstopdf}
\usepackage{tabularx}

\usepackage[normalem]{ulem}%%%%%%%%%

\usepackage{algorithm,algcompatible,amsmath}

\usepackage{hyperref}
\usepackage{url}
\usepackage{bm}
\usepackage{caption}

\usepackage{booktabs}

\hypersetup{colorlinks,citecolor=black,filecolor=black,linkcolor=black,urlcolor=black}
\newtheorem{theorem}{Theorem}
\newtheorem{remark}{Remark}

\newtheorem{assumption}{Assumption}

\newtheorem{lemma}{Lemma}

\title{\LARGE \bf Adaptive MPC for Iterative Tasks}

\author{ Monimoy Bujarbaruah, Xiaojing Zhang, Ugo Rosolia,  Francesco Borrelli  % <-this % stops a space
\thanks{M. Bujarbaruah, X. Zhang, U. Rosolia and F. Borrelli are with the Model Predictive Controls Laboratory at the University of California, Berkeley, USA.
E-mails: \tt\scriptsize{ \{monimoy\_bujarbaruah,  xiaojing.zhang, ugo.rosolia, fborrelli\}@berkeley.edu.}} %
}

\begin{document}

\maketitle
\thispagestyle{empty}
\pagestyle{empty}

%%%%%%%%%%%%%%%%%%%%%%%%%%%%%%%%%%%%%%%%%%%%%%%%%%%%%%%%%%%%%%%%%%%%%%%%%%%%%%%%
\begin{abstract}
This paper proposes an Adaptive Learning Model Predictive Control strategy for uncertain constrained linear systems performing iterative tasks. The additive uncertainty is modeled as the sum of a bounded process noise and an unknown constant offset. As new data becomes available, the proposed algorithm iteratively adapts the believed domain of the unknown offset after each iteration. An MPC strategy robust to all feasible offsets is employed in order to guarantee recursive feasibility. We show that the adaptation of the feasible offset domain reduces conservatism of the proposed strategy, compared to classical robust MPC strategies. As a result, the controller performance improves. Performance is measured in terms of following trajectories with lower associated costs at each iteration. Numerical simulations highlight the main advantages of the proposed approach.
\end{abstract}

%%%%%%%%%%%%%%%%%%%%%%%%%%%%%%%%%%%%%%%%%%%%%%%%%%%%%%%%%%%%%%%%%%%%%%%%%%%%%%%%

\section{Introduction}
Model Predictive Control (MPC) has established itself as a promising tool for dealing with constrained, and possibly uncertain, systems \cite{morari1999model, mayne2000constrained, borrelli2017predictive}. Challenges in MPC design include presence of disturbances and/or unknown model parameters. Disturbances can be handled by means of robust or chance constraints, and such methods are generally well understood \cite{limon2010robust, kothare1996robust, Goulart2006, schwarm1999chance, zhang:margellos:goulart:lygeros:13}. In this paper, we are looking into methods for addressing the second challenge posed by model uncertainties.
 
If the actual model of a system is unknown, adaptive control strategies have been applied for meeting control objectives and ensuring system stability. Literature on unconstrained adaptive control is vast and not the subject of the current paper. Adaptive control for \emph{constrained systems} has mainly focused on improving performance with the adapted models, while the constraints are satisfied robustly for all possible model realizations and for the worst disturbance bounds \cite{aswani2013provably,aswani2012extensions}. In \cite{bujarbaruahm, tanaskovic2014adaptive} a Finite Impulse Response system model is considered for real time adaptation, but this can be  restrictive as it is applicable only for asymptotically stable systems. In \cite{lorenzen2017adaptive}, a state space based Linear Parameter Varying model has been proposed for adaptive MPC with robust constraints. The method modifies the feasible set of states using set translations and scaling \cite{langson2004robust}.
However, this assumes the existence of a unique stabilizing feedback gain matrix for all possible parametric uncertainties, which is potentially a harsh assumption. 
% Moreover, the method is conservative in the way it modifies feasible set of states using set translations and scaling \cite{langson2004robust}. 
In \cite{hewing2017cautious, klenske2016gaussian,  ostafew2014learning}, data driven approaches for learning and adapting system uncertainties are presented with predictive control algorithms, but they do not provide any guarantees of recursive feasibility or stability. 
% Moreover, none of the aforementioned work effectively handles model uncertainty adaptation in an iterative task scenario in presence of constraints. 

In this paper, we tackle the simple, yet insightful problem of regulating a constrained system in presence of  additive model uncertainties. We use a model uncertainty adaptation algorithm inspired by \cite{tanaskovic2014adaptive}, and extend it for systems performing an iterative task. Although control design for systems performing repetitive tasks has been studied \cite{bristow2006survey, carrau2016efficient, rosolia2017learningj}, real-time model adaptation, while guaranteeing satisfaction of state-input constraints, have not yet been addressed thoroughly in literature. 
% adapting the domain set bounds for an additive unknown offset. 

We consider linear time invariant systems represented by a state-space model with known matrices. The system is subject to bounded additive uncertainty, which consists of: $(i)$ a zero mean process noise, that belongs to a known convex set, and $(ii)$ an unknown, but constant offset.
% This decoupling of uncertainty models is often useful in practice, \textbf{e.g, the offset can be a bias acting on the system ???}. 
Given an initial estimate of the offset domain, we iteratively refine it, using a set membership based method \cite{tanaskovic2014adaptive}, as new data becomes available. In order to design an MPC controller with the unknown offset, we make sure the constraints on states and inputs are robustly satisfied for all feasible offsets at a time instant. Here a ``feasible offset" is an offset belonging to the current estimation of the offset domain. 
As the feasible offset domain is updated with data iteratively, we obtain an iterative adaptation in the MPC algorithm. 
% Following the conventional robust MPC methods \cite[Chapter~3]{kouvaritakis2016model}, the system states are decoupled as the sum of a nominal and an error state. 
Moreover, to further improve control performance from one iteration to the next iteration, we use the Learning MPC (LMPC) framework \cite{rosolia2017learningj} to construct the MPC terminal constraints and terminal cost. 
% whereas the \emph{minimal robust positive invariant set} for the corresponding error states is formulated as in \cite{Rosolia2017a}. The latter depends on the size of uncertainty set in the system.
% As the feasible domain for offset is updated with data iteratively, the knowledge of the system uncertainty set is modified, and improved. As a result, the \emph{constraint tightening} \cite[Chapter~3]{kouvaritakis2016model} in our robust MPC controller, which depends on the size of the minimal robust positive invariant set for error states, varies with iterations. This results in iterative adaptation in the algorithm. 

% Enforcement of constraints with these sets ensure recursive feasibility of the resulting controller. 
The main contributions of this paper are summarized as:
\begin{itemize}
    \item We introduce adaptation of additive uncertainty in a robust MPC algorithm, for Linear Time Invariant systems performing iterative tasks. We ensure recursive feasibility and robust stability of the resulting controller. 
    
    \item We show that model adaptation iteratively  \textcolor{black}{relaxes the bounds of the imposed state/input constraints, as more data is collected over time. Due to this relaxation, the optimal control problem in a new iteration can lower the incurred cost further, thus iteratively reducing the \emph{conservatism} in the algorithm.} 
    
    \item We demonstrate performance improvement with our adaptive algorithm, by comparing numerical simulation results to the results of \cite{Rosolia2017a}. 
    %The lower conservatism in our algorithm enables a more aggressive exploration of the state-space, and lower costs at each iteration.

    % \item We merge our adaptive algorithm with the LMPC framework for iterative tasks.   
\end{itemize}

This paper is organized as follows: Section~\ref{sec:ProblemDefinition} introduces the problem setup. Section~\ref{sec:adap_learn_mpc} formulates the Adaptive Learning Model Predictive Control (ALMPC) algorithm basics, outlining the control policy and the solved finite horizon robust optimal control problem, which iteratively improves using  LMPC strategy. We also elaborate the main idea of recursive uncertainty adaptation in this section. 
% Here we also formulate the 
% Section~\ref{sec:model_adap} formulates the model uncertainty adaptation framework. 
% where Section~\ref{ssec:uncer_set_adap} includes the algorithm for recursive model uncertainty adaptation. 
% The ALMPC algorithm is presented in Section~\ref{sec:ALMPC}, where we define the finite horizon robust optimal control problem to be solved.  
% We also formulate the update of disturbance invariant set in Section~\ref{ssec:mrpi_adap} with every iteration and prove recursive feasibility of our algorithm in Section~\ref{ssec:rec_feas}. 
Properties of the proposed ALMPC algorithm are explained in Section~\ref{sec:advant_adap}, where we also prove recursive feasibility and robust stability of the control algorithm.  Section~\ref{sec:example} presents numerical simulations and comparisons, and Section~\ref{sec:conclusion} concludes the paper and outlines possible future extensions.

%%%%%%%%%%%%%%%%%%%%%%%%%%%%%%%%%%%%%%%%%%%%%%%%%%%%%%%%%%%%%%%%%%%%%%%%%%%%%%%%
\section{Problem Statement}\label{sec:ProblemDefinition}
Given an initial state $x_S$, we consider uncertain linear time-invariant systems of the form %% $x_{k+1}=A x_k + B u_k + E w_k$,
\begin{equation}\label{eq:uncSystem}
x_{t+1}=A x_t + B u_t + E\theta_a + w_t,\quad x_0 = x_S,
\end{equation}
where $x_t\in \mathbb{R}^{n_x}$ is the state at time $t$, $u_t\in\mathbb{R}^{n_u}$ is the input, and $A$ and $B$ are known system matrices. At each time step $t$, the system is affected by a random process noise $w_t \in \mathbb W \subset \mathbb{R}^{n_w}$, \textcolor{black}{which is considered to have a zero mean}. In this paper, $\mathbb W$ is assumed to be a compact convex polytope.  % There is also an uncertain, but constant offset $E\theta_a$ in the system 
\textcolor{black}{We also consider the presence of a constant offset $\theta_a \in \mathbb{R}^{p}$, which affects the state dynamics through the known matrix $E \in \mathbb{R}^{n_x \times p}$.} We consider constraints of the form %$F x_t + G u_t \leq f$, 
\begin{equation}\label{eq:constraints_nominal}
	F x_t + G u_t \leq f,
\end{equation}
which must be satisfied for all uncertainty realizations $w_t\in\mathbb{W}$.
% , and all possible offsets $E\theta_t \in \mathbb{D}_t$. 
The matrices $F \in\mathbb{R}^{n_f\times n_x}$, $G \in\mathbb{R}^{n_f \times n_u}$ and $f\in\mathbb{R}^{n_f}$ are assumed known.  
%The goal is to iteratively drive system \eqref{eq:uncSystem} to a neighborhood of the origin. At the $j$-th iteration, let the vector
%\begin{subequations}
%    \begin{align}\label{eq:Statesequence}
%    \mathbf u^{(j)} &:= [u_0^{(j)},\ u_1^{(j)},\ u_2^{(j)},\ \ldots] \\
%    \mathbf x^{(j)} &:= [x_0^{(j)},\ x_1^{(j)},\ x_2^{(j)},\ \ldots],
%\end{align}
%\end{subequations}
%be the collection of inputs applied to system~(\ref{eq:uncSystem}) and the corresponding
%state evolution, . We assume that the closed loop trajectories start from the same initial state $x_S$, (i.e. $x_0^{(j)} ~ = x_S, ~\forall j \geq 0$).
\textcolor{black}{Throughout the paper, we assume that system (\ref{eq:uncSystem}) performs the same task repeatedly. Each task execution is referred to as \emph{iteration}. At each $j^{\textnormal{th}}$ iteration, the performance of the $j^{\textnormal{th}}$ realized closed loop trajectory is quantified with the \emph{iteration cost}
\begin{equation}
    V^{(j)} = \sum_{t= 0}^{\infty} \ell(x_t^{(j)}, u_t^{(j)})
\end{equation}
where $x_t^{(j)}$ and $u_t^{(j)}$ denote the realized system state and the control input at time $t$, respectively, and $\ell: \mathbb{R}^{n_x}\times\mathbb{R}^{n_u} \to \mathbb{R}_+$ is a positive definite stage cost.}
% \textbf{FORMAL DEFINITION OF TASK AND TASK COST}

Our goal is to design a controller that, at each  iteration $j$, solves the infinite horizon robust optimal control problem:
\begin{equation}\label{eq:generalized_InfOCP}
	\begin{array}{clll}
		\hspace{0cm} V^{(j),\star}(x_S) = \\ [1ex]
		\displaystyle\min_{u_0^{(j)},u_1^{(j)}(\cdot),\ldots} & \displaystyle\sum\limits_{t\geq0} \ell \left( \bar{x}_t^{(j)}, u_t^{(j)}\left(\bar{x}_t^{(j)}\right) \right) \\[1ex]
		\text{s.t.}  & x_{t+1}^{(j)} = Ax_t^{(j)} + Bu_t^{(j)}(x_t^{(j)}) + E\theta_a +w_t^{(j)},\\[1ex]
		& Fx_t^{(j)} + Gu_t^{(j)} \leq f,\ \forall w_t^{(j)} \in \mathbb W,\ \\[1ex]
		&  x_0^{(j)} = x_S,\ t=0,1,\ldots,
	\end{array}
\end{equation}
% \textbf{In this problem theta is not defined.. either you are formal  about it (I prefer this) or you say ``we will clarify how theta is handled and modeled later.." }
\textcolor{black}{where $\theta_a$ is the constant offset present in the system}  and $\bar{x}_t^{(j)}$ denotes the disturbance-free nominal state. Notice that \eqref{eq:generalized_InfOCP} minimizes the nominal cost. We point out that, as  system \eqref{eq:uncSystem} is uncertain, the optimal control problem \eqref{eq:generalized_InfOCP} consists of finding $[u_0^{(j)},u_1^{(j)}(\cdot),u_2^{(j)}(\cdot),\ldots]$, where $u_t^{(j)}: \mathbb{R}^{n_x}\ni x_t^{(j)} \mapsto u_t^{(j)} = u_t^{(j)}(x_t^{(j)})\in\mathbb{R}^{n_u}$ are state feedback policies. 
%\textcolor{blue}{As initial realized state $x_0^{(j)}$ is known exactly, hence $u_0^{(j)}$ is not a feedback policy to be computed.}
%\textbf{explain why $u_0$ is not a policy. ALSO if each iteration is infinite time need to explain  how can you get many iterations...}
\textcolor{black}{It is important to note that in practical applications, each iteration has a finite time duration. However, it is common in literature to adopt an infinite time formulation at each iteration, for the sake of simplicity \cite{rosolia2017learningj}, as done in this paper.}

\textcolor{black}{In this paper, we try to compute a solution to the infinite time optimal control problem (\ref{eq:generalized_InfOCP}), by solving a finite time constrained optimal control problem. Moreover, we assume that offset $\theta_a$ in (\ref{eq:generalized_InfOCP}) is not known exactly. Therefore, we propose a parameter estimation framework to refine our knowledge of $\theta_a$ and thus, improve controller performance. We also guarantee recursive constraint satisfaction for all possible offsets.}  

% In the following section, we present the main idea used in this paper to solve (\ref{eq:generalized_InfOCP}) in presence of the unknown offset $\theta_a$. 

%%%%%%%%%%%%%%%%%%%%%%%%%%%%

\section{Adaptive Learning MPC }\label{sec:adap_learn_mpc}

\textcolor{black}{We design a finite horizon MPC controller in an attempt to solve (\ref{eq:generalized_InfOCP})
% (\textbf{ 3 cannot be solved.. not formulated properly}) 
for regulating (\ref{eq:uncSystem}) to the origin.} The key  idea is to use data to learn the unknown offset $\theta_a$, while exploiting the iterative nature of the problem to improve the performance at the next iteration. In particular:
%We elaborate these in the following steps:
\begin{enumerate}

\item The knowledge of the potential domain of offset $\theta_a$ is refined in every iteration, using a set membership based approach \cite{tanaskovic2014adaptive,bujarbaruahm}, giving rise to an Adaptive MPC algorithm. Due to this uncertainty adaptation, \textcolor{black}{and resulting constraint relaxation, the optimal control problem in a new iteration can lower the incurred cost further, thus iteratively reducing \emph{conservatism}.}
% Therefore, the minimal robust positive invariant set of error state is adapted after each iteration, using the updated domain of offset.  

\item We combine the above uncertainty adaptation with a technique called Learning Model Predictive Control (LMPC) \cite{rosolia2017learningj}, to design an \emph{Adaptive Learning MPC} (ALMPC) controller, whose performance improves iteratively by learning the terminal cost and terminal constraint in the MPC algorithm.

% Using ideas from Learning MPC \cite{Rosolia2017a}, we exploit the iterative nature of the task to design controllers whose performance improve over time. 

% We follow a robust MPC design approach along with a Learning LMPC (LMPC) \cite{rosolia2017learningj} algorithm, for robust constraint satisfaction at each time step, for all possible offsets and process noise. The LMPC learns the terminal cost and constraints to yield performance improvement iteratively. 

\end{enumerate}
% Thus, adaptation iteratively lowers constraint tightening for nominal states in tube MPC design, and hence, controller delivers further performance improvements under lesser uncertainties. We also prove recursive feasibility and robust stability of the resulting algorithm.

% \section{Robust LMPC for Uncertain Systems}\label{sec:RLMPC}

% \section{Formulating model uncertainty adaptation}\label{sec:model_adap}

In the following sections, we present the detailed mathematical formulations to achieve steps $(1)$ and $(2)$.

\subsection{Estimating $\theta_a$}\label{ssec:uncer_set_adap}

We characterize the knowledge of the offset $\theta_a$ by its domain $\Theta$, called the \emph{Feasible Parameter Set}, which we estimate from previous system data. 
% The set of all possible offsets $\theta_a$ is called the Feasible Parameter Set \cite{bujarbaruahm}. 
\textcolor{black}{This is initially chosen as a polytope
% an overestimated polytope \textbf{overestimated polytope??? not formal math language !} 
$\Theta^{(0)}$, and it is assumed that $\theta_a \in \Theta^{(0)}$ for all times.} The set is then updated at the end of each iteration after gathering input-output data for all time steps. \textcolor{black}{For} the $j^{\textnormal{th}}$ iteration, the updated \emph{Feasible Parameter Set}, denoted by $\Theta^{(j)}$, is given by
\begin{multline} \label{eq: fps_defin}
    \Theta^{(j)} =  \{\theta \in \mathbb{R}^p: x^{(i)}_{t}-Ax^{(i)}_{t-1}-Bu^{(i)}_{t-1}-E\theta \in \mathbb{W} \\ \forall i \in [0, \ldots,j-1], \forall t \geq 0\}.
\end{multline}
It is clear from (\ref{eq: fps_defin}) that as iterations go on, new data is progressively added to improve the knowledge of $\Theta$, without discarding previous information. Knowledge from all previous iterations is included in $\Theta^{(j)}$. Thus, updated Feasible Parameter Sets are obtained with intersection operations on polytopes, and so $\Theta^{(j+1)} \subseteq \Theta^{(j)}$. 
% We now have the elements to introduce the Adaptive Learning MPC algorithm in the following section.

%%%%%%%%%%%%%%%%%%%%%%%%%%%%%%%%%%%%
\subsection{MPC Problem}\label{sec:ALMPC}
In this section, we present the proposed ALMPC algorithm.  

\subsubsection*{Control Policy Approximation}\label{ssec:con_pol}
We consider affine state feedback policies of the form \cite{bemporad2003min, kouvaritakis2016model}
\begin{equation}\label{eq:inputParam}
	u^{(j)}_t(x_t) = K \left( x^{(j)}_t - \bar{x}^{(j)}_t \right) + v^{(j)}_t,
\end{equation}
where $K\in\mathbb{R}^{n_u \times n_x}$ is a fixed feedback gain same for all iterations $j$. $\bar{x}^{(j)}_t$ is the nominal state and $v^{(j)}_t$ is an auxiliary control input. The parametrization \eqref{eq:inputParam} allows us to decouple the state dynamics \eqref{eq:uncSystem} into a nominal state $s^{(j)}_t \equiv \bar{x}^{(j)}_t$ and an error state $e^{(j)}_t = x^{(j)}_t - s^{(j)}_t$, whose respective dynamics are given by
\begin{subequations}\label{eq:dec_dyn}
    \begin{align}
    	s^{(j)}_{t+1} &= A s^{(j)}_t + Bv_t^{(j)},\quad s^{(j)}_0 = x_S, \quad j=0,1,\ldots  	 \label{eq:s_nominalDyn}\\
    	e^{(j)}_{t+1} &= \Psi e^{(j)}_t + E\theta_a + w_t,~ e^{(j)}_0=0,  \quad j=0,1,\ldots, \label{eq:e_errorDyn}
    \end{align}
\end{subequations}
where $\Psi := (A+BK)$. Above, the state $s_t^{(j)}$ is uncertainty-free, while the error state contains both the process noise as well as the (unknown) offset $\theta_a$.

\subsubsection*{Reformulation of Constraints}
The constraints (\ref{eq:constraints_nominal}) on states and inputs, defined in (\ref{eq:generalized_InfOCP}), can be reformulated in terms of nominal and error states using (\ref{eq:dec_dyn}). Hence, substituting (\ref{eq:dec_dyn}) in (\ref{eq:generalized_InfOCP}), we can derive the corresponding constraints on the nominal states as
\begin{align}
    &Fs_t^{(j)} + Gv_t^{(j)} + (F + GK)e_t^{(j)}  \leq f, \label{eq:state_con}
    \end{align}
where error state $e_t^{(j)}$ follows the dynamics (\ref{eq:e_errorDyn}). Now, we need to ensure that constraints (\ref{eq:state_con}) are satisfied $\forall w_t^{(j)} \in \mathbb W$, in presence of the unknown offset $\theta_a$. Therefore, (\ref{eq:state_con}) are imposed for all models in the Feasible Parameter Set, i.e. $\forall \theta \in {
\Theta}^{(j)}$ and for all $t \geq 0$, in the $j^{\textnormal{th}}$ iteration. Thus, we modify considered error state dynamics (\ref{eq:e_errorDyn}) as 
\begin{align}\label{eq:err_dyn_modif}
e_{t+1}^{(j)} = \Psi e_{t}^{(j)} + d_t^{(j)},\quad e^{(j)}_0=0, ~\quad j=0,1,\ldots,
\end{align}
where $d_t^{(j)}=\{d: d = w + E\theta,\ \forall w \in \mathbb{W},\ \forall \theta \in \Theta^{(j)}$, and ensure
\begin{align*}
    Fs_t^{(j)} + Gv_t^{(j)} + (F + GK)e_t^{(j)}  \leq f,\ \ \forall d_t^{(j)} \in \mathbb{W} \oplus E\Theta^{(j)}. 
\end{align*}
% We model the domain of the additive offset parameter as a bounded polytopic set given by (\ref{eq: fps_defin}). We then, implement a robust MPC algorithm.  %iterative system \eqref{eq:uncSystem}.

\subsubsection*{Tractable Reformulation}
We now solve the following finite horizon robust optimal control problem at each time step $t$ of iteration $j$ 
\begin{equation}\label{eq:rob_LMPC}
	\begin{array}{clll}
		\hspace{0cm} V_{t\to t+N}^{\text{ALMPC},j}(s_t^{(j)}) = \\ [1ex]
		\displaystyle\min_{v_{t|t}^{(j)},\ldots,v_{t+N-1|t}^{(j)}} &  \displaystyle\sum\limits_{k=t}^{t+N-1} \ell \left( s_{k|t}^{(j)}, v_{k|t}^{(j)} \right) + P^{(j-1)}(s_{t+N|t}^{(j)}) \\[1ex]
		\text{s.t.}  & s_{k+1|t}^{(j)} = As_{k|t}^{(j)} + Bv_{k|t}^{(j)},\\[1ex]
% 		& Fs_{k|t}^{(j)} + Gv_{k|t}^{(j)} \leq f - \displaystyle\max_{e\in\mathcal E} \{\Phi e\},\\[1ex]
		& Fs_{k|t}^{(j)} + Gv_{k|t}^{(j)} \leq f - h_s^{(j)}, \\[1ex]
		& s_{t|t}^{(j)}= s_{t}^{(j)},\ s_{t+N|t}^{(j)} \in\mathcal{CS}^{(j-1)},\\[1ex]
		& k=t,\ldots,t+N-1,
	\end{array}
\end{equation}
where $h^{(j)}_s = \max \limits_{e_t \in \mathcal E^{(j)}} \Phi e_t$, $\Phi = (F+GK)$, and the set $\mathcal{E}^{(j)}$ is assumed to be the minimal robust positive invariant set for the error in (\ref{eq:err_dyn_modif}), for iteration $j$ \cite[Definition~3.4]{kouvaritakis2016model}. Terminal constraint $\mathcal{CS}^{(j-1)}$, and the terminal cost $P^{(j-1)}(\cdot)$ are discussed in details next. As $\Theta^{(j)}$ is updated using (\ref{eq: fps_defin}), the set $\mathcal{E}^{(j)}$ alters after every iteration. Therefore, the solved MPC problem  (\ref{eq:rob_LMPC}) is \emph{adaptive} in nature. Upon solving \eqref{eq:rob_LMPC}, the controller applies
\begin{equation}\label{eq:rob_LMPC_ctrLaw}
    u_{t}^{(j)}(x_t^{(j)}) = K(x_t^{(j)} - s_{t}^{(j)}) + v_{t|t}^{(j),\star}
\end{equation}
to the system (\ref{eq:uncSystem}), where $v_{t|t}^{(j),\star}$ is the first input from the optimal input sequence of \eqref{eq:rob_LMPC}. 

Now we discuss the component of \emph{Learning} in this adaptive MPC algorithm, which corresponds to an iterative performance improvement, i.e. decrease of closed loop cost from iteration to iteration. 
%%%%%%%%%%%%%%%

% Due to adaptation algorithm (\ref{eq: fps_defin}), this set also updates after every iteration. 
% The properties are elaborated in the following section.

\subsubsection*{Construction of Terminal Conditions}
In \eqref{eq:rob_LMPC} the nominal terminal state is constrained into the control invariant safe set $\mathcal{CS}^{(j-1)}$, and the terminal cost is denoted by $P^{(j-1)}(\cdot)$. These are borrowed from Learning MPC for deterministic systems \cite{rosolia2017learningj}, whose main idea we revisit next. At iteration $j$, let the vectors
\begin{subequations}\label{eq:sequence}
\begin{align}
 {\bf{v}}^{(j)} ~ = ~ [v_0^{(j)},~v_1^{(j)},~...,~v_t^{(j)},~...], \label{eq:sequenceU} \\
 {\bf{s}}^{(j)} ~ = ~ [s_0^{(j)},~s_1^{(j)},~...,~s_t^{(j)},~...], \label{eq:sequenceX}
\end{align}
\end{subequations}
denote the nominal input and state of system \eqref{eq:s_nominalDyn}. 
%s applied to system~\eqref{eq:dec_dyn} and the corresponding nominal state evolution.
To exploit the iterative nature of the control design, we
define the \emph{sampled Safe Set} $\mathcal{SS}^{(j)}$ at iteration $j$ as
\begin{equation}\label{eq:SS}
\begin{aligned}
\mathcal{SS}^{(j)} = \textrm{}\left\{\bigcup_{i \in M^{(j)}} \bigcup_{t=0}^{\infty} s_t^{(i)} \right\},
\end{aligned}
\end{equation}
where $M^{(j)}\subseteq\{0,\ldots,j\}$ is the set of indices associated with successful iterations, i.e., 
\begin{equation*}\label{eq:M}
 \begin{aligned}
 M^{(j)} = \textrm{} \Big\{ k \in [0,j] : \lim_{t \to \infty} s_t^{(k)} = 0 \Big\}.
 \end{aligned}
 \end{equation*}
In other words, $\mathcal{SS}^{(j)}$ is the collection of all nominal state trajectories up to iteration $j$ that have converged to the origin. We define the \emph{convex Safe Set} as
\begin{equation} \label{eq:CS}
	\begin{aligned}
		\mathcal{CS}^{(j)} = \text{conv}(\mathcal{SS}^{(j)}).
	\end{aligned}
\end{equation}
% \subsubsection{Terminal Cost}

Now, at time $t$ of the $j^{\textnormal{th}}$ iteration, the cost-to-go associated with the closed loop trajectory (\ref{eq:sequenceX}) and input sequence (\ref{eq:sequenceU}) is defined as

\begin{equation*}\label{eq:Functional}
	\begin{aligned}
		V_{t\rightarrow \infty}^{(j)}(s_t^{(j)}) = ~ \sum\limits_{k=0}^{\infty} \ell (s_{t+k}^{(j)},v_{t+k}^{(j)}),
	\end{aligned}
\end{equation*}
where $\ell(\cdot,\cdot)$ is the stage cost of problem (\ref{eq:rob_LMPC}). We define the barycentric function 
\begin{equation}\label{eq:barycentric}
	P^{(j)}(s) = \begin{cases}
		p^{(j),*}(s)  & \mbox{if } s \in \mathcal{CS}^{(j)} \\
		+\infty  & \mbox{else} \\
	\end{cases},
\end{equation} 
where 
\begin{subequations} \label{p^*}
	\begin{align*}
		& p^{(j),*}(s) =  \min_{\lambda_t \geq 0, \forall t \in [0, \infty)}  \sum_{k=0}^{j} \sum_{t=0}^{\infty}\lambda_t^{(k)} V_{t\rightarrow \infty}^{(k)}(s_t^{(k)})  \\
		& \text{s.t.} \qquad\ \sum_{k=0}^{j} \sum_{t=0}^{\infty}\lambda_t^{(k)} = 1,\ \sum_{k=0}^{j} \sum_{t=0}^{\infty}\lambda_t^{(k)} s_t^{(k)}= s, % \\
	\end{align*}
\end{subequations}
where $s_t^{(k)}$ is the nominal state at time $t$ of the $k^{\textnormal{th}}$ iteration, as defined in (\ref{eq:sequenceX}). The function $P^{(j)}(\cdot)$ assigns to every point in $\mathcal{CS}^{(j)}$ the minimum cost-to-go along the trajectories in $\mathcal{CS}^{(j)}$. Therefore this choice of terminal cost in (\ref{eq:rob_LMPC}) guarantees that at each iteration, the iteration cost is non-increasing \cite[Theorem~2]{rosolia2017learningj}.

Intuitively, the function $P^{(j)}(.)$ quantifies the performance of the closed loop trajectories in the previous iterations. This information from the previous iteration data is exploited to guarantee iterative performance improvement and therefore, the algorithm has a ``{learning} property".  
%%%%%%%%%%%%%%%%%%%%%%%%%%%%%%%%%%%%%%%%
\begin{algorithm}
    \caption{
Adaptive Learning MPC
    }
    \label{alg1}
        \begin{algorithmic}[1]
      \STATE Start with iteration $j$. Set $t=0$; initialize Feasible Parameter Set $\Theta^{(j)}$ and compute initial minimal robust positive invariant set $\mathcal E^{(j)}$ using algorithm in \cite[Chapter~3]{kouvaritakis2016model}. 
           
      \STATE Compute $v_{t|t}^{(j),\star}$ from (\ref{eq:rob_LMPC}) and apply $v^{(j)}_t = v_{t|t}^{(j),\star}$ to the system. 

    \STATE 
     Set $t = t+1$, and return to step~2 until the end of $j^{\textnormal{th}}$ iteration.
    
    \STATE At the end of $j^{\textnormal{th}}$ iteration, set the convex hull of the nominal states' as the convex safe set, $\mathcal{CS}^{(j)}$ for next iteration \cite{rosolia2017learningj}.
 
    \STATE Update $\Theta^{(j)}$ using (\ref{eq: fps_defin}). Set $j = j+1$. Return to step~1.
    \end{algorithmic}
\end{algorithm}

We formulate an optimal control problem to find a feasible trajectory which can be used to initialize the ALMPC. We make the following assumption for our system: 

\begin{assumption}\label{ass:rob_LMPC}
We assume we are given an initial state and input sequence $(\mathbf s^{(0)}, \mathbf v^{(0)})$ that satisfies the constraints \eqref{eq:s_nominalDyn} and $Fs_t^{(0)} + Gv_t^{(0)} \leq f - \max_{e\in\mathcal E^{(0)}}\{\Phi e\}$ for all $t\geq0$, where $\mathcal E^{(0)}$ is the minimal robust positive invariant set for \eqref{eq:err_dyn_modif} in the first iteration. 
\end{assumption}
%%%%%%%%%%%%%%%

\remark \label{rem:recfeas}
Since the system \eqref{eq:s_nominalDyn} is linear and the constraints are convex, every convex combination of the trajectories in the sampled safe set $\mathcal{SS}^{(j)}$ is a feasible trajectory, which steers the system to the origin. Therefore $\mathcal{CS}^{(j)}$ is a control invariant set for \eqref{eq:s_nominalDyn} \cite{rosolia2017learning}. This is required for proving recursive feasibility of MPC problem (\ref{eq:rob_LMPC}) in closed loop with controller (\ref{eq:rob_LMPC_ctrLaw}) \cite[Theorem~1]{rosolia2017learning}. 

\begin{remark}
\textcolor{black}{The decoupling of the state dynamics in (\ref{eq:dec_dyn}) is a linear change of coordinates, and is not necessary to implement the proposed strategy. Indeed, it is possible to reformulate the controller in the coordinate frame of system (\ref{eq:uncSystem}). In particular, using standard set theory tools \cite[Chaper 10]{borrelli2017predictive}, one  could tighten the constraints (\ref{eq:constraints_nominal}) and the terminal set $\mathcal{CS}$ (\ref{eq:CS}) to guarantee robustness against all disturbance realizations}.
\end{remark}
% \textbf{ADD a remark stating that if you had not used the error model, the set CS should have been "robustified".. Ugo knows what i am talking...}
 
%%%%%%%%%%%%%%%%%%%%%%%%%%%%%%

\section{Properties of Adaptation}\label{sec:advant_adap}
In this section we discuss 
the properties and associated advantages of model uncertainty adaptation (\ref{eq: fps_defin}), namely: $(i)$ decreasing constraint tightening, and $(ii)$ convergence guarantees.  

\subsection{Decreased Constraint Tightening}\label{ssec:mrpi_adap}
It is evident that in (\ref{eq:rob_LMPC}), the constraints on the nominal state $s$ are tightened based on the size of the set $\mathcal{E}^{(j)}$. The set $\mathcal E^{(j)}\subset\mathbb{R}^{n_x}$ the minimal robust positive invariant set for the error dynamics \eqref{eq:err_dyn_modif} in iteration $j$, which is defined as \cite[Eq.~3.23]{kouvaritakis2016model} 
\begin{align}\label{eq:mrpi_set}
\mathcal{E}^{(j)} = \bigoplus \limits_{i=0}^{\infty} \Psi^i(\mathbb{W} \oplus E\Theta^{(j)}),
\end{align}
which satisfies the condition that, for all $e_0^{(j)} \in \mathcal E^{(j)}$, we have $e^{(j)}_{t}(\mathbf w_{t-1}, \textcolor{black}{\theta})\in\mathcal E^{(j)}$ for all $\mathbf w_{t-1}=[w_0,\ldots,w_{t-1}]\in\mathbb W$, for all $\theta \in \Theta^{(j)}$ and all $t\geq1$. 
% In this section we discuss two key aspects of the above robust positive invariant set, namely: $(i)$ adaptation with data and size adjustment, and $(ii)$ real-time computability and convergence.
% \subsubsection*{Adaptation}
% where $\rho$ and $r$ satisfy $\Phi^r(\mathbb{W} \oplus E\Theta^{(j)}) \subseteq \rho(\mathbb{W} \oplus E\Theta^{(j)})$ \cite{kouvaritakis2016model}
\begin{lemma} \label{lem:invar}
The minimal robust positive invariant set $\mathcal{E}^{(j)}$, is monotonically decreasing with respect to iterations. That is, for successive iterations of the system, it follows the property $\mathcal{E}^{(j+1)} \subseteq \mathcal{E}^{(j)}$. 
\end{lemma}
\begin{proof}
New cuts are added to the Feasible Parameter Set polytope $\Theta^{(j)}$ according to (\ref{eq: fps_defin}) to obtain $\Theta^{(j+1)}$ at the next iteration. Namely, from (\ref{eq: fps_defin}), we have that
\begin{multline*}
    \Theta^{(j+1)} = \Theta^{(j)} \cap   \{\theta \in \mathbb{R}^p: x^{(j+1)}_{t}-Ax^{(j+1)}_{t-1}-Bu^{(j+1)}_{t-1}\\
    -E\theta \in \mathbb{W},  \forall t \geq 0\}.
\end{multline*}
Therefore, $\Theta^{(j+1)} \subseteq \Theta^{(j)}$. So from (\ref{eq:mrpi_set}) we get
\begin{multline}\label{eq:mrpi_set_next}
\mathcal{E}^{(j+1)} = \bigoplus \limits_{i=0}^{\infty} \Psi^i(\mathbb{W} \oplus E\Theta^{(j+1)}) \subseteq   \bigoplus \limits_{i=0}^{\infty} \Psi^i(\mathbb{W} \oplus E\Theta^{(j)}).
\end{multline}
This shows that $\mathcal{E}^{(j+1)} \subseteq \mathcal{E}^{(j)}$, and concludes the proof.
\end{proof}
The decreasing size of the minimal robust positive invariant set $\mathcal{E}^{(j)}$, clearly reduces constraint tightening in (\ref{eq:rob_LMPC}) for nominal states. That indicates a less conservative algorithm as more data becomes available. Now we prove recursive feasibility of the resulting ALMPC algorithm, for which Lemma~\ref{lem:invar} is a sufficient condition along with terminal constraints given by convex safe set $\mathcal{CS}$.

\begin{theorem}
Now consider system \eqref{eq:uncSystem} in closed loop with the ALMPC \eqref{eq:rob_LMPC} and \eqref{eq:rob_LMPC_ctrLaw}. Let Assumption~\ref{ass:rob_LMPC} and (\ref{eq:mrpi_set_next}) hold. Then  ALMPC \eqref{eq:rob_LMPC}, \eqref{eq:rob_LMPC_ctrLaw}  is feasible for all times $t\geq0$ and all iterations $j\geq1$. 
\end{theorem}

\begin{proof} As shown in Lemma 1, $\mathcal{E}^{(j+1)} \subseteq \mathcal{E}^{(j)}$, this implies that 
\begin{equation}
    \max_{e \in \mathcal{E}^{(j+1)}} \Phi e \leq \max_{e \in \mathcal{E}^{(j)}} \Phi e.
\end{equation}
Consequently, if a realized trajectory at iteration $j$ robustly satisfies input and state constraints,
\begin{equation}
    Fs_{k|t}^{(j)} + Gv_{k|t}^{(j)} \leq f - h^{(j)}_s,\ h^{(j)}_s = \max \limits_{\forall e_t \in \mathcal E^{(j)}} \Phi e_t,
\end{equation}
then this trajectory would also robustly satisfy state and input constraints at the next iteration
\begin{equation}
    Fs_{k|t}^{(j+1)} + Gv_{k|t}^{(j+1)} \leq f - h^{(j+1)}_s,\ h^{(j+1)}_s = \max \limits_{\forall e_t \in \mathcal E^{(j+1)}} \Phi e_t.
\end{equation}
This implies that at each $j^{\textnormal{th}}$, iteration the first $j-1$ iterations robustly satisfy state and input constraints for system (\ref{eq:uncSystem}) and the disturbance $(\mathbb{W} \oplus E{\Theta}^{(j)})$. Therefore, using the fact from Remark~\ref{rem:recfeas} and \cite[Theorem 1]{rosolia2017learning}, the ALMPC is recursively feasible at each $j^{\textnormal{th}}$ iteration.
% \textcolor{blue}{Note that the system dynamics in (\ref{eq:s_nominalDyn}) and \eqref{eq:rob_LMPC} are identical. Moreover, the optimization problem \eqref{eq:rob_LMPC} is deterministic. Therefore recursive feasibility of ALMPC \eqref{eq:rob_LMPC}, \eqref{eq:rob_LMPC_ctrLaw} follows from \cite[Theorem 1]{rosolia2017learning}. \textbf{Write more. Use the above advantage. Conclusively prove feasibility here. Talk to Ugo and finalize this section.}} 
\end{proof}

\subsection{Convergence Guarantees}
In this section we first show that system \eqref{eq:uncSystem} in closed loop with the controller (\ref{eq:rob_LMPC_ctrLaw}), guarantees convergence to a neighborhood of the origin. Then, we show that the adaptive strategy allows at each iteration to shrink the guaranteed region of convergence around the origin.

\begin{theorem}\label{thm:RLMPC_stab}
Consider system \eqref{eq:uncSystem} in closed loop with the ALMPC \eqref{eq:rob_LMPC} and \eqref{eq:rob_LMPC_ctrLaw}. Let Assumption~\ref{ass:rob_LMPC} hold. Then the closed-loop system \eqref{eq:uncSystem}, \eqref{eq:rob_LMPC_ctrLaw} converges asymptotically to the set $\mathcal E^{j}$ for every iteration $j\geq1$ and all $w\in\mathbb W$, i.e $x_t^{(j)} \rightarrow \mathcal{E}^{(j)}$ as $t\rightarrow \infty$. 
\end{theorem}

\begin{proof}
The proof is based on \cite[Theorem 1]{rosolia2017learning}: We have that $\lim_{t \rightarrow \infty} s_t^{(j)}=0, ~ \forall j \geq 1$ \cite{rosolia2017learning}. Finally, from the definition of $e_t = x_t^{(j)} - s_t^{(j)}$ we have $\lim_{t \rightarrow \infty} x_t^{(j)} = \lim_{t \rightarrow \infty} (s_t^{(j)} + e_t^{(j)})=\lim_{t \rightarrow \infty} e_t^{(j)} \in \mathcal{E}^{(j)}, ~ \forall j \geq 1$, which concludes the proof. 
\end{proof}

Now, from (\ref{eq:mrpi_set_next}), we know $\mathcal{E}^{(j+1)} \subseteq \mathcal{E}^{(j)}$. Hence, the ALMPC delivers improved guarantees of the region of convergence iteratively, as a consequence of update (\ref{eq: fps_defin}).

%%%%%%%%%%%%%%%%%%%%%%%%
\section{Simulation Results}\label{sec:example}

In this section we compare the results of our ALMPC algorithm in Algorithm~1 with the Robust Learning MPC (RLMPC) presented in \cite{Rosolia2017a}. For numerical simulations, we apply both ALMPC and RLMPC algorithms to compute feasible solutions to the following iterative infinite horizon optimal  control problem
\begin{equation}\label{eq:generalized_InfOCP_ex}
	\begin{array}{llll}
		\hspace{-0.cm}    & V^{(j),\star}(x_S)  =
		\displaystyle\min_{u_0^{(j)},u_1^{(j)}(\cdot),\ldots}  \displaystyle\sum\limits_{t\geq0} \left \| \bar{x}_t^{(j)} \right\|_1 + 10 \left\| u_t^{(j)} \right\|_1  \\[1ex]
		& \text{s.t.}\\
		&x_{t+1}^{(j)} = \begin{bmatrix} 1.2 & 1.5\\
		0 & 1.3\end{bmatrix}x_t^{(j)} + \begin{bmatrix} 0\\
		1\end{bmatrix}u_t^{(j)}(x_t^{(j)}) \\
		& ~~~~~~~~~~~~~~~~~~~~~~~~~~~~~~~+ \begin{bmatrix} 1 & 0\\ 0 & 1\end{bmatrix}\begin{bmatrix} 0.01 \\ 0.05\end{bmatrix} +  w_t^{(j)},\\[2.5ex]
		& \begin{bmatrix}-10 \\ -10 \\ -1
		\end{bmatrix} \leq \begin{bmatrix}x_t^{(j)} \\ u_t^{(j)}
		\end{bmatrix} \leq \begin{bmatrix}10 \\ 10 \\ 1
		\end{bmatrix},\ \forall w_t^{(j)} \in \mathbb W, \\[3.5ex]
		&  x_0^{(j)} = x_S,\ t=0,1,\ldots,
	\end{array}
\end{equation}
where $w_t \in \mathbb W = \{w \in \mathbb R^2: ||w||_{\infty} \leq 0.8 \}$. The initial Feasible Parameter Set is defined as 
\begin{align}\label{eq:init_FPS}
    \Theta^{(0)} = \{ \theta \in \mathbb{R}^2: \begin{bmatrix}-0.2 \\ -0.1 \\
	\end{bmatrix} \leq \theta \leq \begin{bmatrix}0.2 \\ 0.1
	\end{bmatrix}\}.  
\end{align}
The true offset parameter is $\theta_a = [0.01, 0.05]^\top$. The matrix $E \in \mathbb{R}^{2 \times 2}$ is picked as the identity matrix. We solve the above optimization problem for arbitrarily chosen initial state $x_S = [-5.6,1.29]^\top$. The ALMPC in \eqref{eq:rob_LMPC}, \eqref{eq:rob_LMPC_ctrLaw}, and the RLMPC \cite{Rosolia2017a} are implemented with a control horizon of $N=3$, and the feedback gain $K$ in \eqref{eq:rob_LMPC_ctrLaw} is chosen to be the optimal LQR gain for system \eqref{eq:s_nominalDyn} with parameters $Q=I$ and $R = 10$. 
% A feasible initial trajectory to initialize $\mathcal{CS}^0$ and $P^0(\cdot)$  was found by applying the rigid tube MPC method of \cite{kouvaritakis2016model}. 
% For each iteration $j$, the RLMPC schemes was terminated  at step $t$ whenever $V_{t\rightarrow t+ N}^{\scalebox{0.6}{RLMPC},(j)}(s_t^{(j)}) \leq 10^{-8}$.

In the case of RLMPC, the domain set for the unknown offset is kept constant throughout as $\Theta^{(0)}$. Contrarily, the ALMPC algorithm proposed in this paper, shrinks the domain of the offset set with (\ref{eq: fps_defin}). This adaptation lowers the degree of constraint tightening over iterations due to (\ref{eq:mrpi_set_next}).
% , and thus, lowers the conservatism of the algorithm.

\subsection{Uncertainty Adaptation in ALMPC}
We illustrate the adaptation of additive uncertainty in this section. Feasible Parameter Set is adapted after every iteration according to (\ref{eq: fps_defin}) after being initialized as (\ref{eq:init_FPS}). This is seen from Fig.~\ref{Fig:theta_ad}. It is clearly observed from Fig.~\ref{Fig:theta_ad} that the additive offset uncertainty is shrunk with every iteration. The true offset, which is picked solely for simulation purposes, is marked by \textcolor{blue}{$\star$} and is always inside the Feasible Parameter Set for all iterations. This underscores the fact that the robustification of the MPC problem (\ref{eq:rob_LMPC}) for all feasible offsets in the Feasible Parameter Set, also captures the dynamics with the true offset.     
\begin{figure}[h!]
    \centering
	\includegraphics[width= 1.05\columnwidth]{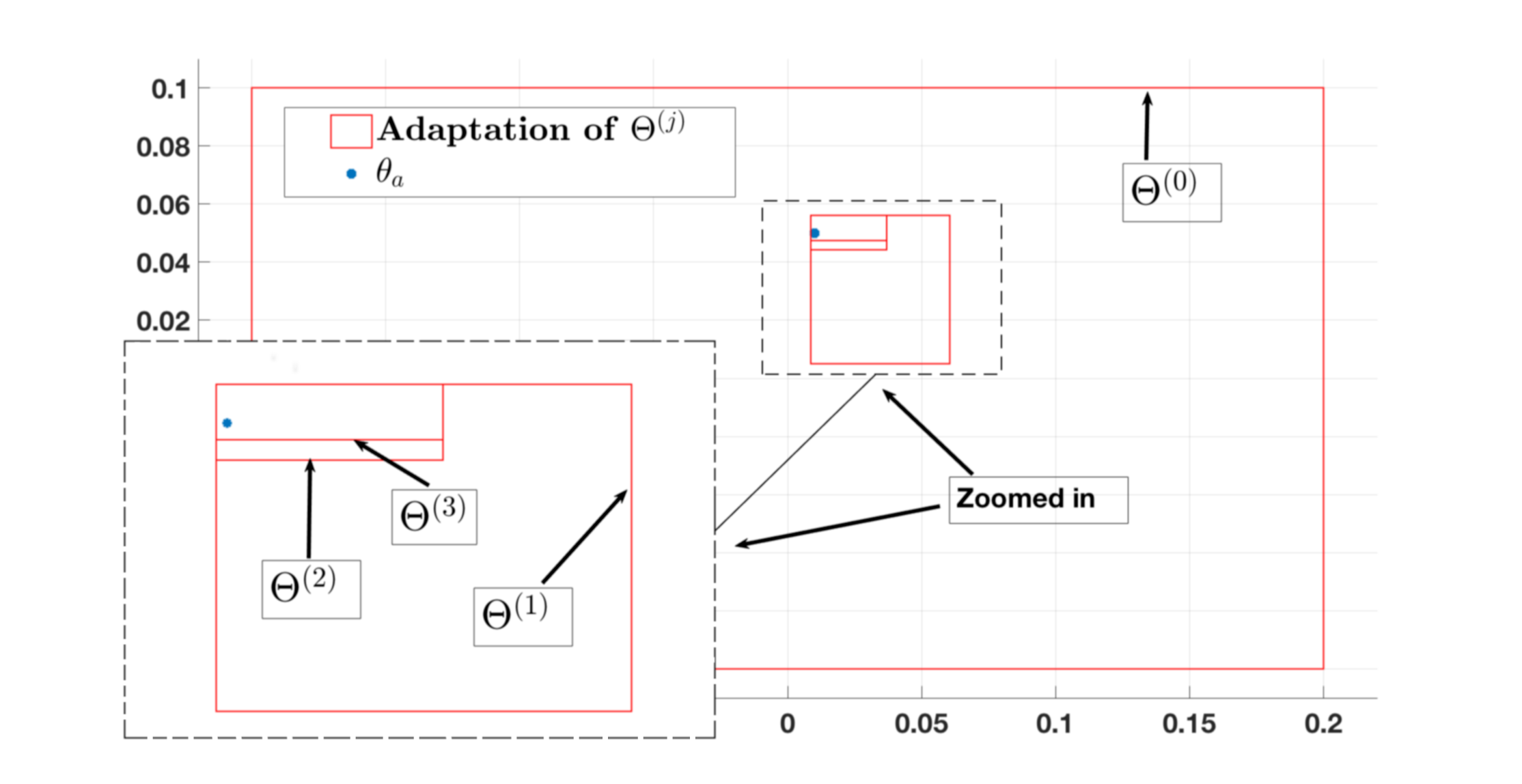}
    \caption{Adaptation of Feasible Parameter Set}
    \label{Fig:theta_ad}
\end{figure}

As a consequence of the above uncertainty adaptation, the size of the minimal robust positive invariant set $\mathcal{E}^{(j)}$ shrinks as in (\ref{eq:mrpi_set_next}). This is highlighted in Fig.~\ref{Fig:invar_ad}. Whereas for RLMPC, the size of the minimal robust positive invariant set for dynamics (\ref{eq:err_dyn_modif}) is constant. This keeps the same magnitude of constraint tightening throughout all iterations in (\ref{eq:rob_LMPC}).

\begin{figure}[h!]
    \centering
	\includegraphics[width= \columnwidth]{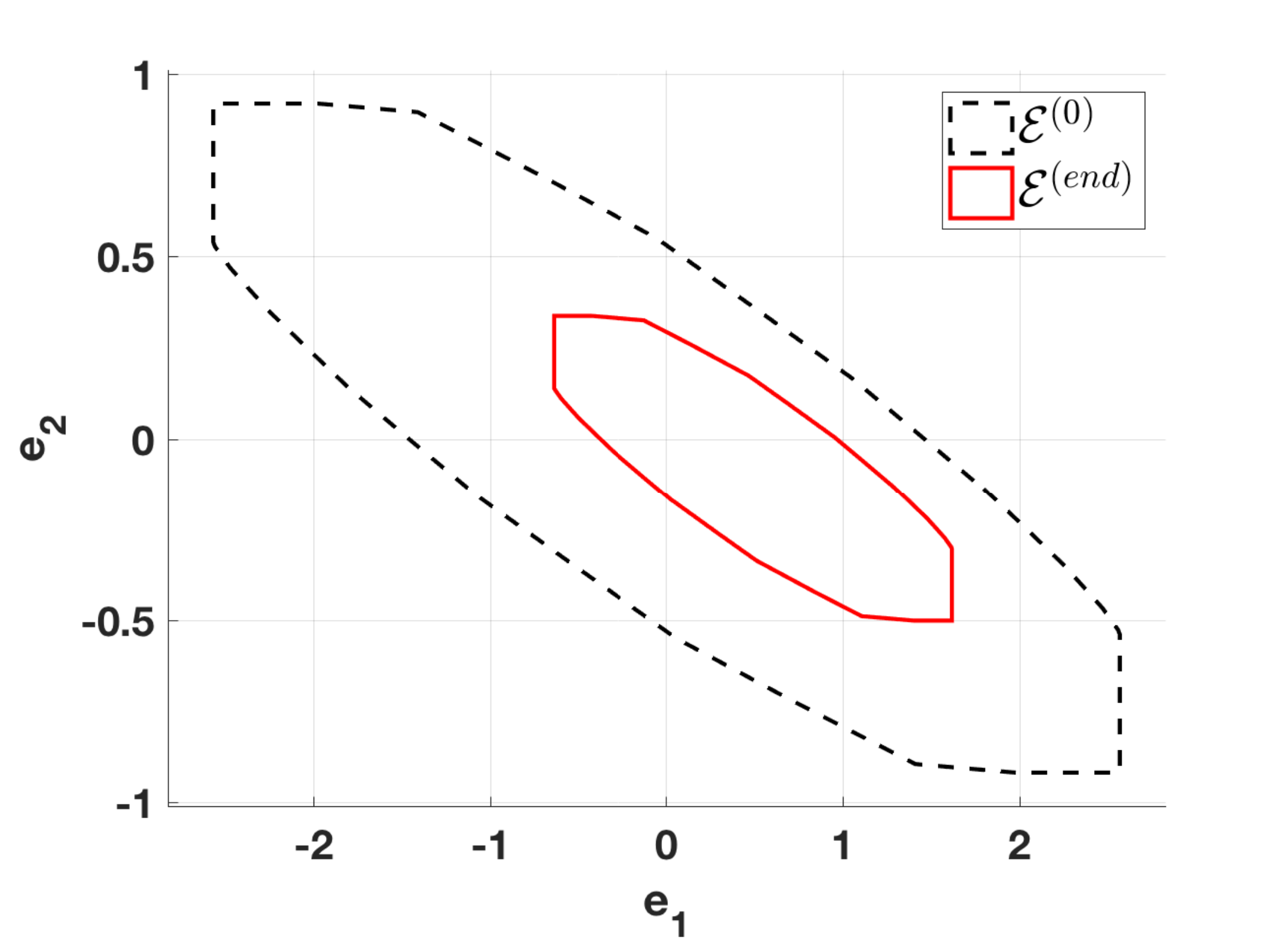}
    \caption{Adaptation of Minimal Robust Positive Invariant Set}
    \label{Fig:invar_ad}
\end{figure}

\subsection{Performance Gain Over RLMPC} In this section we compare the results of the ALMPC algorithm with RLMPC \cite{Rosolia2017a} to illustrate the performance improvement as a result of uncertainty adaptation, and thus, minimal robust positive invariant set adaptations as shown in Fig.~\ref{Fig:theta_ad} and Fig.~\ref{Fig:invar_ad}. We quantify this performance gain with two quantifiers:
\begin{itemize}
\item Trajectories with enhanced exploration of state space
\item Lower iteration costs
% \item Increased size of convex safe set for terminal constraint of nominal states
% \item Improved robust stability \textcolor{blue}{add a graph. talk to Ugo.} 
\end{itemize}

\begin{figure}[h!]
    \centering
	\includegraphics[width= \columnwidth]{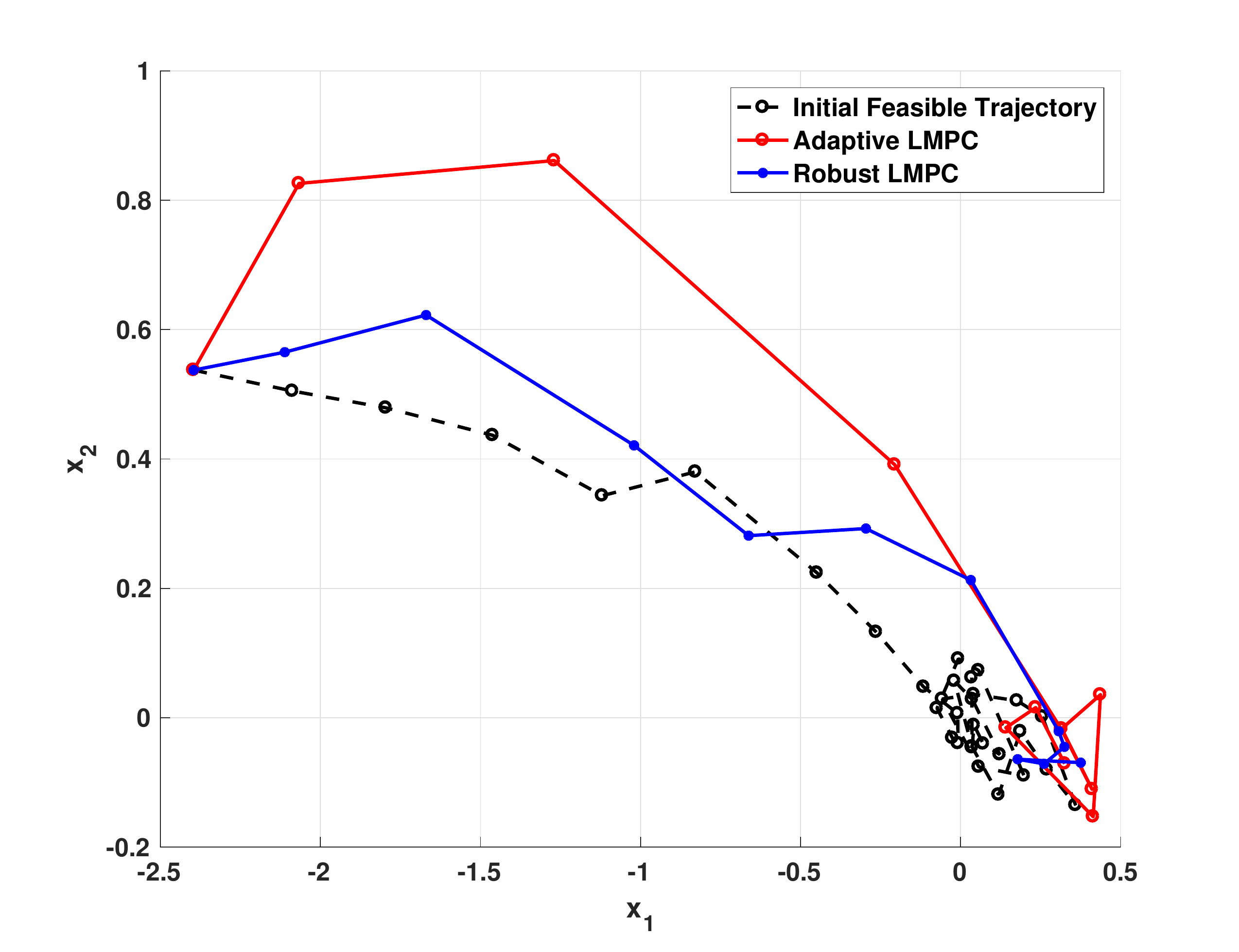}
    \caption{\textcolor{black}{Trajectory Comparison after Fifth Iteration}}
    \label{Fig:traj_comp}
\end{figure}

\begin{figure}[h!]
    \centering
	\includegraphics[width= \columnwidth]{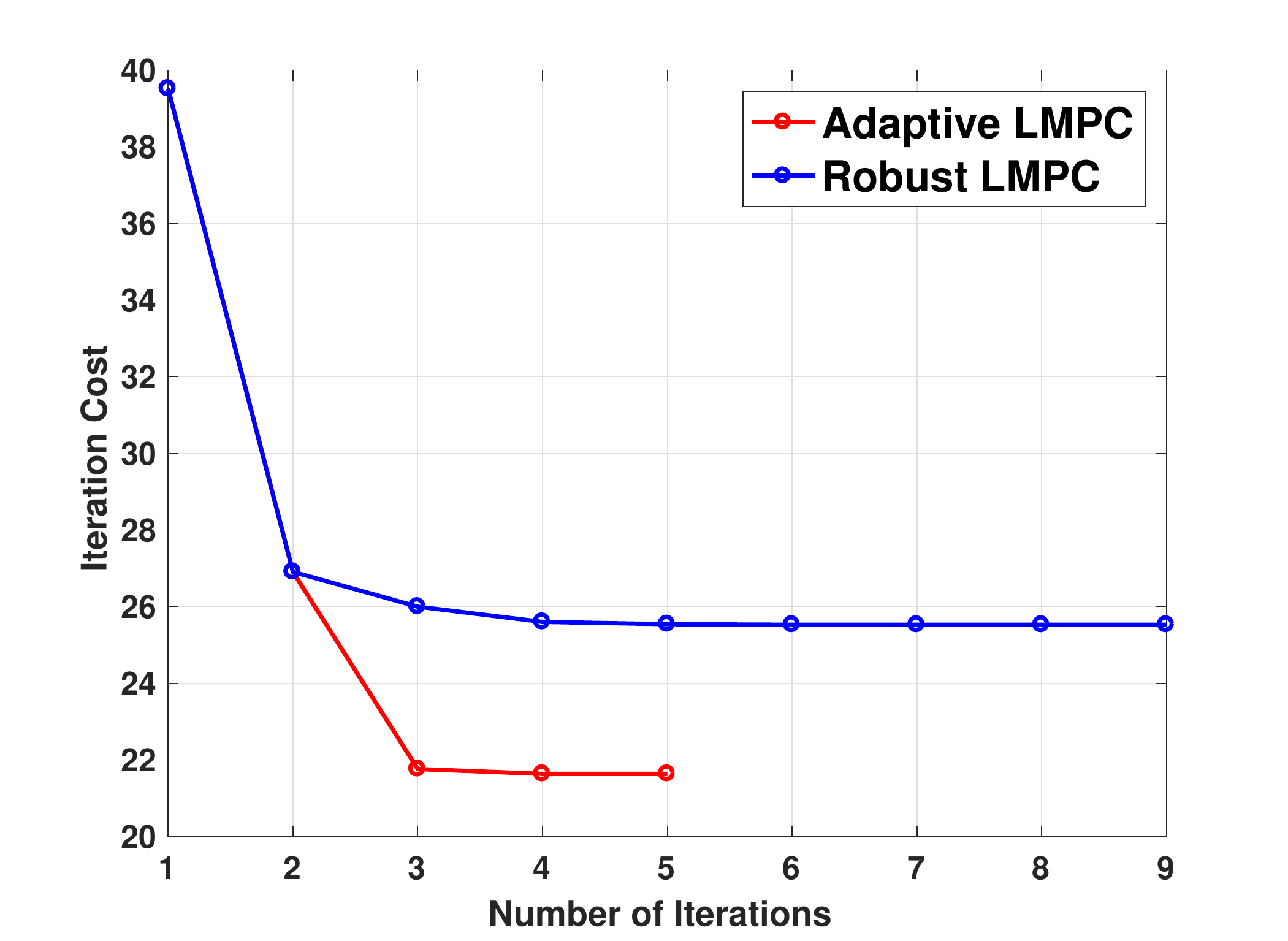}
    \caption{Iteration Cost Comparison}
    \label{Fig:iter_cost}
\end{figure}

% \begin{figure}[h!]
%     \centering
% 	\includegraphics[width= \columnwidth]{figures/CS_Comparison.eps}
%     \caption{Convex Safe Set Comparison}
%     \label{Fig:CSS_comp}
% \end{figure}
%%%%%%%%%%%%%%%%%%%%%%%%%%%%%%%%%
Fig.~\ref{Fig:traj_comp} shows the comparison of final system trajectories with both the algorithms. It is observed that the ALMPC algorithm generates a converged trajectory farther away from the initialized one, highlighting better exploration property. On the other hand, the RLMPC algorithm trajectory stays relatively closer to the initialized path, as it overestimates potential uncertainties by not adapting them. 

To prove that this aforementioned exploration advantage does not come at the penalty of additional iteration costs, we compare the iteration costs of the two algorithms in Fig.~\ref{Fig:iter_cost}. The ALMPC algorithm results in lower iteration costs, again proving the advantage of uncertainty adaptation.

% \textcolor{blue}{Finally in Fig.~ show improved robust stability properties. In RLMPC final convergence is to a set which is bigger, in ALMPC that set of convergence shrinks as $\mathcal{E}^{(j)}$ shrinks.} 
% for ALMPC.

\balance

%%%%%%%%%%%%%%%%%%%%%%
\section{Conclusion and future work}\label{sec:conclusion}
In this paper, we successfully merged an additive uncertainty adaptation methodology with the LMPC framework, to formulate the ALMPC algorithm. We proved that with every iteration, if the uncertainty is learned and thus shrunk, the amount of constraint tightening is lowered for nominal states. We have conclusively shown with numerical simulations, that the ALMPC framework, compared to RLMPC \cite{Rosolia2017a}, explores better trajectories with lower iteration costs as uncertainty is recursively adapted. 
% Moreover, the convex safe set for nominal states is larger with adaptation, again underscoring enhanced exploration abilities. 
Thus, ALMPC is shown to have improved performance over RLMPC. We also proved recursive feasibility and robust stability of the ALMPC algorithm. 

In future extensions of this work, we aim to solve an iterative  stochastic MPC problem with model uncertainty adaptation. The inherent uncertainty adaptation algorithm in presence of constraints, can also be exploited in systems as \cite{bujarbaruahlyap, bujarbaruahtorque} for obtaining better performance. Finally, we also wish to deal with time varying uncertainties in state-space, building on the work of \cite{2017arXiv171207548T}. 
% This would provide us a deeper understanding of model learning in controls, with theoretical feasibility and stability guarantees. 
%%%%%%%%%%%%%%%%%%%%%
%%%%APPENDIX%%%%%

% \section*{Appendix}

% \section*{ACKNOWLEDGMENT}
% The authors would like to ***

%%%%%%%%%%%%%%%%%%%%%%%%%%%%%%%%%%%%%%%%%%%%%%%%%%%%%%%%%%%%%%%%%%%%%%%%%%%%%%%%

\renewcommand{\baselinestretch}{0.930}
% saman commented out
% \bibliographystyle{IEEEtran}
% \bibliography{IEEEabrv,acc2018.bib} 
\renewcommand{\baselinestretch}{1}

% saman added
\section*{References}
{%\footnotesize 
\printbibliography[heading=none, resetnumbers=true]
}

\end{document}